\pdfoutput=1
\documentclass[aps,prl,reprint,superscriptaddress,longbibliography,nobibnotes]{revtex4-2}
\usepackage{amsmath,amssymb,bm,mathtools}
\usepackage{amsthm}
\usepackage{mathrsfs}
\usepackage{graphicx}
\usepackage{placeins}
\usepackage{booktabs}
\usepackage{microtype}
\usepackage{hyperref}
\hypersetup{
  colorlinks=true,
  allcolors=blue
}
\theoremstyle{definition}
\newtheorem{proposition}{Proposition}
\newtheorem{lemma}{Lemma}
\newtheorem{corollary}{Corollary}
\newcommand{\dd}{\mathrm{d}}
\newcommand{\Ree}{\operatorname{Re}}
\newcommand{\sgnfun}{\operatorname{sgn}}
\newcommand{\Rrat}{\mathcal R}
\newcommand{\Pclass}{\mathfrak P}
\newcommand{\Ephys}{\mathcal E}
\newcommand{\Nmin}{N_{\min}}
\newcommand{\Wclass}{\mathfrak W}

\newcommand{\paperTitle}{Coupling Does Not Reduce the Auxiliary-Mode Count for \texorpdfstring{$1/|\omega|$}{1/|omega|} Spectra in Passive Lindblad Networks}

\begin{document}

\title{\paperTitle}

\author{Qing-Ao Xiang}
\affiliation{Key Laboratory of Opto-electronic Control and Detection Technology of University of Hunan Province, and College of Physics and Electronic Engineering, Hengyang Normal University, Hengyang 421002, China}

\author{Yan Liu}
\affiliation{Key Laboratory of Opto-electronic Control and Detection Technology of University of Hunan Province, and College of Physics and Electronic Engineering, Hengyang Normal University, Hengyang 421002, China}

\author{Xin-Yuan Yang}
\affiliation{Key Laboratory of Opto-electronic Control and Detection Technology of University of Hunan Province, and College of Physics and Electronic Engineering, Hengyang Normal University, Hengyang 421002, China}

\author{Ya-Ju Song}
\email{yjsong@hynu.edu.cn}
\affiliation{Key Laboratory of Opto-electronic Control and Detection Technology of University of Hunan Province, and College of Physics and Electronic Engineering, Hengyang Normal University, Hengyang 421002, China}
\affiliation{Key Laboratory of Low Dimensional Quantum Structures and Quantum Control of Ministry of Education, Hunan Normal University, Changsha 410081, China}

\begin{abstract}
Representing continuous environments by finitely many Markovian auxiliary modes is fundamental in non-Markovian open quantum systems, yet a critical question remains: at a fixed mode budget, can coherent intermode coupling reduce the spectral approximation error? We prove that intermode coupling offers no advantage when passive, number-conserving Gaussian Lindblad auxiliary networks approximate a $1/|\omega|$ spectrum over a finite two-sided frequency band. For any mode budget $N$, the general coupled class and its uncoupled diagonal subclass share the same optimal error, which is exactly the degree-$2N$ Zolotarev error for sign approximation. This optimum is attainable by $N$ independent damped auxiliary modes at zero detuning. The result holds when the auxiliary network is in a stationary vacuum state, the system couples to it via a single Hermitian bath operator, and no white-noise feedthrough term is present. Consequently, although a general coupled network has $O(N^{2})$ real parameters, coherent intermode coupling, collective dissipation, and nonnormal structure cannot reduce the number of auxiliary modes required to reach a prescribed tolerance. This exact relation yields both the minimum mode count for a prescribed positive-frequency dynamic range and tolerance, and the maximum dynamic range attainable for a prescribed mode budget and tolerance.
\end{abstract}

\maketitle

Representing continuous environments by finitely many Markovian auxiliary modes is common in non-Markovian open quantum systems, quantum impurity problems, and noise modeling; related embeddings are also used in generalized Langevin dynamics~\cite{Baczewski2013}. At fixed mode budget, two questions arise: what is the minimum error compatible with the finite-mode structure, and can this optimum be attained by completely positive and trace-preserving (CPTP) Lindblad dynamics~\cite{Gorini1976,Lindblad1976}? Pseudomode and related finite-mode methods commonly use independent auxiliary modes or sums of exponentials~\cite{Garraway1997,Mazzola2009,Tamascelli2018,Pleasance2020,Trivedi2021}. Coupled Lindblad constructions add coherent intermode coupling and collective dissipation while keeping the joint dynamics CPTP~\cite{Mascherpa2020,Dorda2014,Dorda2015,Lednev2024,Huang2026}. With $N$ modes, a general coupled network has $O(N^2)$ real parameters and can produce interference, dispersive structures, and, through Jordan blocks of a nondiagonalizable drift matrix, higher-order poles~\cite{Alford2026}.

Huang \emph{et al.} connected coupled Lindblad and quasi-Lindblad pseudomode representations and developed robust constructions avoiding nonconvex optimization~\cite{Huang2026,Park2024}. Together with related results, these works show that the mode count needed to fit a bath correlation function over a dimensionless time window $T>0$ to tolerance $\varepsilon_{\rm corr}>0$ can scale as $\operatorname{polylog}(T/\varepsilon_{\rm corr})$~\cite{Thoenniss2025,VilkoviskiyAbanin2024,HuangDecomp2026}. These works primarily address how many modes suffice; we ask whether, at fixed mode budget, intermode coupling reduces the maximum relative spectral error with respect to $1/|\omega|$ over the entire target band.

Low-frequency $1/f$-type noise is an important source of decoherence in solid-state quantum devices and can be characterized over broad frequency ranges by dynamical-decoupling noise spectroscopy~\cite{Yoshihara2006,Bylander2011}. Its multiple frequency scales make it a natural benchmark for finite-mode environment representations. Sums of Lorentzians with different relaxation scales can synthesize finite-band $1/f$-type spectra~\cite{Machlup1954,SauveSzabo1985,Pellegrini1986,Milotti1995}, while Zolotarev theory supplies the relevant minimax rational approximants~\cite{Akhiezer1992,IstaceThiran1995,Gawlik2019,TrefethenWilber2025}. Using the rational-degree bound imposed by the mode count, we derive a Zolotarev lower bound for the general coupled class and show that it is attained within the uncoupled diagonal subclass by $N$ independent damped auxiliary modes at zero detuning. Hence the general coupled class and its uncoupled diagonal subclass have the same optimal error at fixed mode budget. This exact relation yields both the minimum mode count for a prescribed positive-frequency dynamic range and tolerance and the maximum positive-frequency dynamic range for a prescribed mode budget and tolerance; see Fig.~\ref{fig:main}.

\begin{figure*}[t]
\centering
\includegraphics[width=0.97\textwidth]{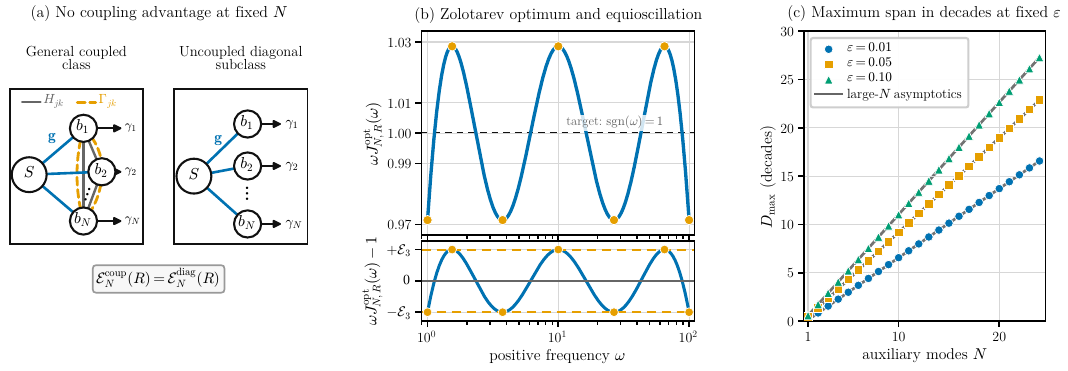}
\caption{
\textbf{No advantage from coupling at fixed mode budget and the resulting resource relations.}
(a) General coupled class versus uncoupled diagonal subclass; $S$ denotes the system. Blue lines $\bm g$ couple the system to the auxiliary modes, gray lines $H_{jk}$ denote coherent intermode coupling, orange dashed lines $\Gamma_{jk}$ denote cross damping from collective dissipation, and black arrows denote individual damping rates $\gamma_j$, where $j,k=1,\ldots,N$.
(b) For $R=10^2$ and $N=3$, $\omega J_{N,R}^{\rm opt}(\omega)$ equioscillates about $\sgnfun(\omega)=1$ on the positive-frequency band. Orange dots mark alternation points at which the residual $\omega J_{N,R}^{\rm opt}(\omega)-1$ alternately attains $\pm\Ephys_3(10^2)$, where $\Ephys_3(10^2)=0.02856$, certifying minimax optimality.
(c) Maximum log-frequency span $D_{\max}(N,\varepsilon)=\log_{10}R_{\max}(N,\varepsilon)$ at fixed tolerance, in decades. Colored markers show exact inversion results; gray lines show the large-$N$ asymptotic forms.
}
\label{fig:main}
\end{figure*}

\paragraph{Model, spectrum, and optimization problem.---}
To compare the general coupled class and the uncoupled diagonal subclass at the same mode budget, we first nondimensionalize the target spectrum and target band. Suppose that the physical target spectrum is
$J_{\rm tar}^{\rm phys}(\nu)=\mathcal A_0/|\nu|$,
where $\mathcal A_0>0$, and that the target band satisfies
$0<\nu_{\min}<\nu_{\max}$ and
$\nu_{\min}\le|\nu|\le\nu_{\max}$.
Set $\omega=\nu/\nu_{\min}$ and measure spectral amplitudes in units of
$\mathcal A_0/\nu_{\min}$. The normalized problem then depends only on the dimensionless positive-frequency dynamic range
$R=\nu_{\max}/\nu_{\min}$:
\begin{equation}
J_{\rm tar}(\omega)=\frac{1}{|\omega|},
\ \ 
\Omega_R=[-R,-1]\cup[1,R],
\ \ R>1 .
\label{eq:target}
\end{equation}
Here $R$ is the positive-frequency dynamic range. We exclude $(-1,1)$ to avoid the zero-frequency singularity. The~optimization is over $\Omega_R$; no condition is imposed outside the band. In a quantum realization, $\omega$ may be interpreted as a detuning in the interaction picture. The even target $1/|\omega|$ can describe the two-sided power spectral density of stationary real-valued classical noise; in the quantum setting, it serves as a scalar spectral benchmark spanning multiple scales in the interaction picture. The spectrum of a generic equilibrium quantum bath obeys detailed balance, and its positive- and negative-frequency branches are generally asymmetric and must be treated separately.

Let $n$ denote the actual number of modes in a given network and $N$ the allowed mode budget. Collect the annihilation operators of the $n$ auxiliary modes into the vector
$\hat{\bm b}=(\hat b_1,\ldots,\hat b_n)^T$.
We consider passive, number-conserving Gaussian auxiliary networks. Here passive means that the Heisenberg-picture linear drift dynamics closes on $\hat{\bm b}$: the linear drift term in
$\dot{\hat{\bm b}}$ contains $\hat{\bm b}$ but not
$\hat{\bm b}^{\dagger}$, so no Bogoliubov mixing occurs.
The matrices $H=(H_{jk})_{j,k=1}^{n}=H^\dagger$ and
$\Gamma=(\Gamma_{jk})_{j,k=1}^{n}
=\Gamma^\dagger\succeq0$
denote, respectively, the number-conserving single-particle Hamiltonian matrix and the damping matrix. For
$j\neq k$, the off-diagonal entry $H_{jk}$ represents coherent coupling between modes $j$ and $k$, whereas $\Gamma_{jk}$ represents cross damping induced by collective dissipation
~\cite{Huang2026,Alford2026}.

The system couples to the auxiliary network through a single Hermitian bath operator:
\begin{equation}
\hat H_{SA}=\hat S\otimes\hat B,
\qquad
\hat B=\bm g^\dagger\hat{\bm b}
+\hat{\bm b}^\dagger\bm g ,
\end{equation}
where $\hat S=\hat S^\dagger$ is a system operator and $\bm g$ is the coupling vector between the system and the auxiliary modes. Consequently, the spectrum seen by the system is scalar. This restriction applies only to the interface through which the system couples to the network and does not restrict the internal Markovian dissipation of the auxiliary network, which remains described by an arbitrary $\Gamma\succeq0$.
With our convention for damping rates, the homogeneous evolution of the auxiliary-mode first moments obeys $\dd\langle\hat{\bm b}\rangle/\dd t=-(\Gamma+iH)\langle\hat{\bm b}\rangle$. We therefore define the drift matrix $M=\Gamma+iH$, and let $I_n$ denote the $n\times n$ identity matrix. Under the above nondimensionalization, time is measured in units of $\nu_{\min}^{-1}$, $H$, $\Gamma$, and $M$ in units of $\nu_{\min}$, and $\bm g$ in units of $\sqrt{\mathcal A_0}$. We allow $M$ to be nonnormal, i.e.,
$MM^\dagger\neq M^\dagger M$. We require every eigenvalue of $M$ to have a positive real part. This stability condition ensures that $e^{-Mt}$ decays with time and excludes poles on the real-frequency axis.

\FloatBarrier
For the zero-mean Gaussian auxiliary networks considered here, under linear system--bath coupling, the bath statistics seen by the system are completely determined by the correlation function of the bath operator $\hat B$~\cite{Tamascelli2018}. Upon expanding $\hat B$, all operator averages vanish in a stationary vacuum state except the annihilation--creation correlator
$\langle\hat{\bm b}(t)\hat{\bm b}^{\dagger}(0)\rangle_0$.
The quantum regression theorem states that, for $t\ge0$, this correlator follows the same homogeneous evolution $e^{-Mt}$ as the single-time expectation values~\cite{Lax1963}. Stationarity and $\hat B=\hat B^\dagger$ then give
\begin{equation}
\begin{aligned}
C(t):&=\langle\hat B(t)\hat B(0)\rangle_0
=\bm g^\dagger e^{-Mt}\bm g,\qquad t\ge0,\\
C(-t)&=C(t)^*,\quad t\ge0,\quad
J(\omega):=\int_{-\infty}^{\infty}C(t)e^{i\omega t}\dd t,\\
J(\omega)&=2\Ree\!\left[\bm g^\dagger(M-i\omega I_n)^{-1}\bm g\right]\ge0 .
\end{aligned}
\label{eq:model-spectrum}
\end{equation}
Proposition~\ref{prop:positivity} in the Supplemental Material proves that this spectrum is nonnegative and has the high-frequency decay
$J(\omega)=O(\omega^{-2})$. To convert the actual mode count $n$ into a rational-degree constraint on the spectrum seen by the system, define
\begin{equation}
\begin{aligned}
h(\omega):=\bm g^\dagger(M-i\omega I_n)^{-1}\bm g
=\frac{u(\omega)}{v(\omega)},\\
u,v\in\mathbb C[\omega],\ \ 
\deg u\le n-1,\ \ \deg v=n .
\end{aligned}
\end{equation}
The rational function $h(\omega)=u(\omega)/v(\omega)$ follows from the matrix inversion formula and does not require $M$ to be diagonalizable. Combining Eq.~\eqref{eq:model-spectrum} with
$J(\omega)=O(\omega^{-2})$ shows that, after cancellation, the numerator and denominator degrees of $J(\omega)$ are at most $2n-2$ and $2n$, respectively. Thus, $J(\omega)$ is of type at most
$(2n-2,2n)$, and $\omega J(\omega)$ is of type at most
$(2n-1,2n)$. We refer to these bounds on the numerator and denominator degrees, imposed by the mode count $n$, as the \emph{rational-degree bound}.
Coherent intermode coupling, collective dissipation, and nonnormal structure can alter the detailed pole structure and spectral coefficients but cannot violate this rational-degree bound. A proof is given in Lemma~\ref{lem:degree} of the Supplemental Material.

Let $\Pclass_{\le N}^{\rm coup}$ denote the set of all triples
$(H,\Gamma,\bm g)$ with actual mode count
$n\le N$ that satisfy the above conditions; we call it the \emph{general coupled class}.
Its \emph{uncoupled diagonal subclass}
$\Pclass_{\le N}^{\rm diag}$ further requires
$H=\operatorname{diag}(\Omega_1,\ldots,\Omega_n)$ and
$\Gamma=\operatorname{diag}(\gamma_1,\ldots,\gamma_n)$, where $\Omega_j$ and $\gamma_j>0$ are, respectively, the center frequency of auxiliary mode $j$ (its center detuning in the chosen interaction picture) and its damping rate. Substitution of these diagonal forms into
Eq.~\eqref{eq:model-spectrum} gives
\begin{equation}
J_{\rm diag}(\omega)=
2\sum_{j=1}^{n}\frac{|g_j|^2\gamma_j}
{(\omega-\Omega_j)^2+\gamma_j^2},
\end{equation}
so the uncoupled diagonal subclass produces a sum of Lorentzians with nonnegative weights.
By contrast, the general coupled class allows $H$ and $\Gamma$ to have off-diagonal entries and also allows the drift matrix $M$ to be nondiagonalizable. Figure~\ref{fig:main}(a) summarizes the structural distinction between the two model classes. Both classes allow $\bm g=0$ and hence contain the zero spectrum, but neither includes a frequency-independent white-noise feedthrough term. Here a feedthrough term means an additional constant spectral term
$d>0$, i.e.,
$J(\omega)\mapsto d+J(\omega)$.

Once the actual mode count $n$ and the parameters $H$, $\Gamma$, and $\bm g$ are fixed, they specify a particular network and hence its model spectrum $J(\omega)$. Since
$J_{\rm tar}(\omega)=1/|\omega|$, the relative error of this network with respect to the target spectrum at frequency
$\omega$ is
\begin{equation}
\left|\frac{J(\omega)}{J_{\rm tar}(\omega)}-1\right|
=
|1-|\omega|J(\omega)|.
\end{equation}
We call the maximum of this pointwise error across the entire target band the \emph{maximum relative error over the target band}. Taking its infimum over the corresponding model class gives the optimal error of that class at a prescribed mode budget. Thus,
\begin{equation}
\begin{aligned}
\Ephys_N^{\rm coup}(R)
&=\inf_{(H,\Gamma,\bm g)\in\Pclass_{\le N}^{\rm coup}}
\sup_{\omega\in\Omega_R}|1-|\omega|J(\omega)|,\\
\Ephys_N^{\rm diag}(R)
&=\inf_{(H,\Gamma,\bm g)\in\Pclass_{\le N}^{\rm diag}}
\sup_{\omega\in\Omega_R}|1-|\omega|J(\omega)|,
\end{aligned}
\label{eq:physical-error}
\end{equation}
with $\Ephys_0^{\rm coup}(R)=\Ephys_0^{\rm diag}(R)=1$.
Proposition~\ref{prop:positive-weight-duality} in the Supplemental Material shows that, for any given network, its maximum relative error over the target band equals the supremum, over all continuous nonnegative frequency weights on
$\Omega_R$ that are not identically zero, of the relative deviation of the corresponding weighted spectral quantity. Filter functions in quantum decoherence and noise spectroscopy are physical examples of such weights
~\cite{Cywinski2008,AlvarezSuter2011,Cerfontaine2021}.

\paragraph{From the physical spectrum to the fourth Zolotarev problem.---}
To obtain a lower bound valid for the entire general coupled class, we do not optimize the network parameters
$H,\Gamma,\bm g$ directly. Instead, we use the rational-degree bound above to relax the physical optimization problem by enlarging the admissible set to the corresponding class of real rational functions. For any
$I\subset\mathbb R$ and nonnegative integers $a,b$, let
$\Rrat_{a,b}^{\mathbb R}(I)$ denote the class of real rational functions with no poles on $I$. After cancellation, the numerator and denominator degrees of each function in this class are at most $a$ and $b$, respectively. Any network with actual mode count $n$ generates a spectrum for which
\begin{equation}
\begin{aligned}
q (\omega) := \omega J (\omega) \in \mathcal{R} _ {2 n - 1, 2 n} ^ {\mathbb{R}} \left(\Omega_ {R}\right),\\ 
| 1 - | \omega | J (\omega) | = | \operatorname{sgn} (\omega) - q (\omega) |.
\end{aligned}
\label{eq:sign-lift}
\end{equation}
Equation~\eqref{eq:sign-lift} shows that
$q(\omega)\approx+1$ is required on the positive-frequency band and
$q(\omega)\approx-1$ on the negative-frequency band. The relative spectral approximation problem on the original two-sided band is therefore equivalent to sign-function approximation on two disjoint intervals, with the maximum relative error over the target band unchanged.

To adopt the standard notation for the classical Zolotarev problem, set
$\xi=\omega/R$ and $k=R^{-1}$. Under this scaling, the physical band
$\Omega_R$ maps to
$I_k=[-1,-k]\cup[k,1]$, while the rational degrees and error remain unchanged:
\begin{equation}
\sup_{\omega\in\Omega_R}
|\sgnfun(\omega)-q(\omega)|
=
\sup_{\xi\in I_k}
|\sgnfun(\xi)-q(R\xi)|.
\end{equation}

The fourth Zolotarev problem asks for a real rational function of prescribed type that approximates the sign function on these two disjoint intervals while minimizing the maximum error over their union~\cite{Akhiezer1992,IstaceThiran1995}. For a nonnegative integer $m$, its optimal error is
\begin{equation}
e_m(k)=
\inf_{s\in\Rrat_{m,m}^{\mathbb R}(I_k)}
\sup_{\xi\in I_k}|\sgnfun(\xi)-s(\xi)| .
\label{eq:zolo-error}
\end{equation}
We call a rational function that attains this error an \emph{optimal rational approximant}.

To fold both frequency bands onto one positive interval, we use an odd optimal rational approximant. Oddness is not imposed \emph{a priori} on the physical network: $J(\omega)$ need not be even, so $q(\omega)=\omega J(\omega)$ need not be odd. Nevertheless, the optimal error over the full type-$(2N,2N)$ class can be attained by an odd function that, after cancellation, has type $(2N-1,2N)$ and satisfies the rational-degree bound~\cite{Gawlik2019}. Factoring out its argument leaves a rational function of the squared argument. Scaling back to the physical band gives
\begin{equation}
q_{N,R}^{\rm opt}(\omega)
=\omega r_{N,R}^{\rm opt}(\omega^2),\ \ 
r_{N,R}^{\rm opt}
\in\Rrat_{N-1,N}^{\mathbb R}([1,R^2]).
\label{eq:sign-inverse-factorization}
\end{equation}
Here $q_{N,R}^{\rm opt}(\omega)$ is the rational function that approximates
$\sgnfun(\omega)$ on the two-sided band $\Omega_R$ and minimizes the maximum error over that band, whereas
$r_{N,R}^{\rm opt}(x)$ is the real rational function that approximates
$x^{-1/2}$ on the positive interval $[1,R^2]$ and minimizes the maximum relative error over that interval.
They are related by Eq.~\eqref{eq:sign-inverse-factorization}. The subscripts $N$ and $R$ denote the mode budget and positive-frequency dynamic range, respectively, and the superscript
${\rm opt}$ denotes optimality.

Setting $x=\omega^2$ folds both frequency bands onto
$[1,R^2]$, and $|\sgnfun(\omega)-q_{N,R}^{\rm opt}(\omega)|=|1-\sqrt{x}\,r_{N,R}^{\rm opt}(x)|$. Thus, sign-function approximation on the two-sided band is equivalent to minimax approximation of
$x^{-1/2}$ in relative error on the positive interval, and the two problems have the same optimal error.

The Zolotarev problem therefore supplies a lower bound for the general coupled class. If
$r_{N,R}^{\rm opt}(x)$ has negative real poles and positive residues, this bound is attainable within the uncoupled diagonal subclass.

\textbf{Theorem: No advantage from coupling at fixed mode budget.}
\label{thm:main}
For any $R>1$ and integer $N\ge1$,
\begin{equation}
\boxed{
\begin{aligned}
\Ephys_N^{\rm coup}(R)
&=\Ephys_N^{\rm diag}(R)
=e_{2N}(R^{-1})\\
&=\min_{r\in\Rrat_{N-1,N}^{\mathbb R}([1,R^2])}
\max_{x\in[1,R^2]}
|1-\sqrt{x}\,r(x)| .
\end{aligned}}
\label{eq:main-theorem}
\end{equation}
An optimal rational approximant $r_{N,R}^{\rm opt}(x)$ can be chosen in the following partial-fraction form. Returning to the frequency variable
$\omega$ then defines the corresponding spectrum
$J_{N,R}^{\rm opt}(\omega)$:
\begin{equation}
\begin{aligned}
r_{N,R}^{\rm opt}(x)
&=\sum_{j=1}^{N}\frac{a_j}{x+\gamma_j^2},
\ \  a_j,\gamma_j>0,\\
J_{N,R}^{\rm opt}(\omega)
&:=r_{N,R}^{\rm opt}(\omega^2)
=\sum_{j=1}^{N}
\frac{a_j}{\omega^2+\gamma_j^2}
=2\sum_{j=1}^{N}
\frac{|g_j|^2\gamma_j}
{\omega^2+\gamma_j^2}.
\end{aligned}
\label{eq:positive-stieltjes}
\end{equation}
This spectrum is realized by $N$ independent damped auxiliary modes
at zero detuning:
\begin{equation}
H_{\rm opt}=0,\ \ 
\Gamma_{\rm opt}=\operatorname{diag}(\gamma_1,\ldots,\gamma_N),\ \ 
|g_j|^2=\frac{a_j}{2\gamma_j}.
\label{eq:optimal-physical-realization-main}
\end{equation}

Equation~\eqref{eq:positive-stieltjes} also gives the correspondence between the mathematical approximation parameters and the physical network parameters. In the $x$ variable,
$-\gamma_j^2$ is a negative real pole of $r_{N,R}^{\rm opt}(x)$ and
$a_j>0$ is the corresponding partial-fraction residue. Returning to the frequency variable $\omega$,
$\gamma_j>0$ becomes the damping rate of auxiliary mode $j$, while
$|g_j|^2=a_j/(2\gamma_j)$
is the squared coupling amplitude between the system and that mode. Thus, each partial-fraction term in
$r_{N,R}^{\rm opt}(x)$ corresponds to a physically admissible passive auxiliary mode with damping rate $\gamma_j$ and zero detuning. Figure~\ref{fig:main}(b) shows the equioscillation structure of
$\omega J_{N,R}^{\rm opt}(\omega)$.

\paragraph{Proof outline.---}

\emph{(i) Global lower bound.---}
By Eq.~\eqref{eq:sign-lift} and the scaling above, the function
$q(R\xi)$ corresponding to any network in
$\Pclass_{\le N}^{\rm coup}$ belongs to
$\Rrat_{2N-1,2N}^{\mathbb R}(I_{R^{-1}})$, which is contained in the full class
$\Rrat_{2N,2N}^{\mathbb R}(I_{R^{-1}})$ used in the fourth Zolotarev problem. Enlarging the admissible class can only lower the minimax error, so
$\Ephys_N^{\rm coup}(R)\ge e_{2N}(R^{-1})$.

\emph{(ii) Physical realization.---}
The optimal error for the degree-$2N$ instance of the fourth Zolotarev problem is attained by an odd rational function that, after cancellation, has type
$(2N-1,2N)$. The pole--zero interlacing established in Sec.~\ref{sec:SM-realization} of the Supplemental Material, together with the residue-sign calculation in Eq.~\eqref{eq:SM-residues}, shows that all partial-fraction residues in Eq.~\eqref{eq:positive-stieltjes} are positive. Hence
$J_{N,R}^{\rm opt}(\omega)$ can be constructed within the uncoupled diagonal subclass, and therefore
$\Ephys_N^{\rm diag}(R)\le e_{2N}(R^{-1})$.

\emph{(iii) Inclusion of model classes.---}
Because
$\Pclass_{\le N}^{\rm diag}\subseteq\Pclass_{\le N}^{\rm coup}$,
taking the infimum over the larger model class cannot increase the optimal error, so
$\Ephys_N^{\rm coup}(R)\le\Ephys_N^{\rm diag}(R)$.
Combining the three steps gives
\begin{equation}
e_{2N}(R^{-1})
\le\Ephys_N^{\rm coup}(R)
\le\Ephys_N^{\rm diag}(R)
\le e_{2N}(R^{-1}),
\label{eq:inequality-chain}
\end{equation}
so every inequality is saturated. Rigorous proofs of the degree bound, the sign--inverse-square-root isometry, and the positive-residue realization are given in Secs.~\ref{sec:SM-degree}--\ref{sec:SM-realization} of the Supplemental Material.

We henceforth denote the common optimal error of the general coupled class and the uncoupled diagonal subclass by
$\Ephys_N(R)$. Although the uncoupled diagonal subclass allows nonzero center detunings, the realization at zero detuning already attains this optimal error. The simplest case, $N=1$, is $\gamma_1=\sqrt R$, $\Ephys_1(R)=\left(\frac{\sqrt R-1}{\sqrt R+1}\right)^2$.
Here $\gamma_1=\sqrt R$ is the geometric mean of the two band-edge frequencies $1$ and $R$, making the relative errors at the two band edges equal in magnitude. The higher-order Zolotarev construction is a multimode extension of this minimax balancing principle.

\paragraph{From the exact error to mode count and positive-frequency dynamic range.---}
The minimum number of auxiliary modes required at fixed tolerance follows by inverting the main theorem. Define
\begin{equation}
\begin{aligned}
K(k)
&=\int_0^{\pi/2}
\frac{\dd\theta}
{\sqrt{1-k^2\sin^2\theta}},\\
\mu(k)
&=\frac{\pi}{2}
\frac{K(\sqrt{1-k^2})}{K(k)} .
\end{aligned}
\label{eq:mu-main}
\end{equation}
Here $K(k)$ is the complete elliptic integral of the first kind. The function $\mu(k)$ is continuous and strictly decreasing on $(0,1)$. In terms of $\mu$, the relation among positive-frequency dynamic range, rational degree, and optimal error can be written in a form that can be inverted exactly. The classical Zolotarev formula implies that there is a unique
$\zeta_{N,R}\in(0,1)$ satisfying~\cite{Akhiezer1992,IstaceThiran1995,Gawlik2019}
\begin{equation}
2N\mu(\zeta_{N,R})=\mu(R^{-1}),\qquad
\Ephys_N(R)
=\frac{1-\zeta_{N,R}}{1+\zeta_{N,R}}.
\label{eq:zeta-error-main}
\end{equation}
For $\varepsilon\ge0$, define
\begin{equation}
\Nmin(R,\varepsilon):=
\min\!\left\{
N\in\mathbb N_0:
\Ephys_N(R)\le\varepsilon
\right\},
\label{eq:Nmin-def-main}
\end{equation}
where $\mathbb N_0=\{0,1,2,\ldots\}$ and, by convention,
$\Ephys_0(R)=1$. If the set is empty, we set
$\Nmin(R,\varepsilon)=+\infty$.
By definition, $\Nmin(R,\varepsilon)$ is the minimum number of auxiliary modes required to reach tolerance
$\varepsilon$. For
$0<\varepsilon<1$, the condition
$\Ephys_N(R)\le\varepsilon$ is equivalent to
$\zeta_{N,R}\ge(1-\varepsilon)/(1+\varepsilon)$.
Using the strict monotonicity of $\mu$ and taking the smallest integer $N$ that satisfies this condition gives
\begin{equation}
\boxed{
\Nmin(R,\varepsilon)=
\left\lceil
\frac{\mu(R^{-1})}
{2\mu((1-\varepsilon)/(1+\varepsilon))}
\right\rceil .}
\label{eq:Nmin-main}
\end{equation}
Here $\lceil\cdot\rceil$ is the ceiling function. In the joint asymptotic limit
$R\to\infty$ and $\varepsilon\downarrow0$,
\begin{equation}
\Nmin(R,\varepsilon)=
\frac{\ln(4R)\ln(4/\varepsilon)}{\pi^2}
[1+o(1)] .
\label{eq:Nmin-asymptotic}
\end{equation}
Boundary cases and the exact finite-$N$ error parametrization are given in Corollary~\ref{cor:SM-resource} of the Supplemental Material.

The exact error relation can also be inverted to obtain the maximum positive-frequency dynamic range attainable for a prescribed mode budget and tolerance,
$R_{\max}(N,\varepsilon)=\sup\{R>1:\Ephys_N(R)\le\varepsilon\}$,
and the corresponding log-frequency span
$D_{\max}(N,\varepsilon)=\log_{10}R_{\max}(N,\varepsilon)$.
For fixed $0<\varepsilon<1$ and $N\to\infty$,
\begin{equation}
D_{\max}(N,\varepsilon)=
\frac{
2\mu((1-\varepsilon)/(1+\varepsilon))
}{\ln10}N
-\log_{10}4+o(1) .
\label{eq:Dmax-asymptotic}
\end{equation}
For tolerances of $1\%$, $5\%$, and $10\%$, each additional auxiliary mode asymptotically increases the maximum log-frequency span by approximately
$0.72$, $0.98$, and $1.16$ decades, respectively; see Fig.~\ref{fig:main}(c).

\paragraph{Conclusions and implications.---}
The mode count fixes the rational-degree bound on the spectrum seen by the system and hence the accuracy attainable at a prescribed positive-frequency dynamic range. Although a general coupled network has $O(N^2)$ real parameters, coherent intermode coupling, collective dissipation, and nonnormal structure only change pole locations, multiplicities, and spectral coefficients within this bound; they neither lower the optimal error nor reduce the mode count required at a prescribed tolerance. This exact relation provides an optimal benchmark for finite-band $1/f$-type environment representations, pseudomode design, and multiscale noise synthesis throughout the general coupled class. With spectral estimates or simultaneous confidence bands, it also provides a lower bound on the auxiliary-mode count required to describe an experimental noise spectrum within this model class.

\begin{acknowledgments}
This work was supported by the National Natural Science Foundation of China (12205088), the Natural Science Foundation of Hunan Province (2026JJ50350, 2025JJ50005), the Scientific Research Fund of Hunan Provincial Education Department of China (24C0353), the Open Project of Key Laboratory of Low Dimensional Quantum Structures and Quantum Control of Ministry of Education of Hunan Normal University (QSQC2602), and the Open Project of Key Laboratory of Opto-electronic Control and Detection Technology of University of Hunan Province (2024HSKFJJ012).
\end{acknowledgments}

\bibliographystyle{apsrev4-2}
\bibliography{refs}

@article{Gorini1976,
  author  = {Gorini, Vittorio and Kossakowski, Andrzej and Sudarshan, E. C. G.},
  title   = {Completely Positive Dynamical Semigroups of {$N$}-Level Systems},
  journal = {J. Math. Phys.},
  volume  = {17},
  number  = {5},
  pages   = {821--825},
  year    = {1976},
  doi     = {10.1063/1.522979}
}

@article{Lindblad1976,
  author  = {Lindblad, G.},
  title   = {On the Generators of Quantum Dynamical Semigroups},
  journal = {Commun. Math. Phys.},
  volume  = {48},
  number  = {2},
  pages   = {119--130},
  year    = {1976},
  doi     = {10.1007/BF01608499}
}

@article{Garraway1997,
  author  = {Garraway, B. M.},
  title   = {Nonperturbative Decay of an Atomic System in a Cavity},
  journal = {Phys. Rev. A},
  volume  = {55},
  number  = {3},
  pages   = {2290--2303},
  year    = {1997},
  doi     = {10.1103/PhysRevA.55.2290}
}

@article{Mazzola2009,
  author  = {Mazzola, L. and Maniscalco, S. and Piilo, J. and Suominen, K.-A. and Garraway, B. M.},
  title   = {Pseudomodes as an Effective Description of Memory: {Non-Markovian} Dynamics of Two-State Systems in Structured Reservoirs},
  journal = {Phys. Rev. A},
  volume  = {80},
  number  = {1},
  pages   = {012104},
  year    = {2009},
  doi     = {10.1103/PhysRevA.80.012104}
}

@article{Tamascelli2018,
  author  = {Tamascelli, D. and Smirne, A. and Huelga, S. F. and Plenio, M. B.},
  title   = {Nonperturbative Treatment of {Non-Markovian} Dynamics of Open Quantum Systems},
  journal = {Phys. Rev. Lett.},
  volume  = {120},
  number  = {3},
  pages   = {030402},
  year    = {2018},
  doi     = {10.1103/PhysRevLett.120.030402}
}

@article{Pleasance2020,
  author  = {Pleasance, Graeme and Garraway, B. M. and Petruccione, Francesco},
  title   = {Generalized Theory of Pseudomodes for Exact Descriptions of {Non-Markovian} Quantum Processes},
  journal = {Phys. Rev. Research},
  volume  = {2},
  number  = {4},
  pages   = {043058},
  year    = {2020},
  doi     = {10.1103/PhysRevResearch.2.043058}
}

@article{Trivedi2021,
  author  = {Trivedi, Rahul and Malz, Daniel and Cirac, J. Ignacio},
  title   = {Convergence Guarantees for Discrete Mode Approximations to {Non-Markovian} Quantum Baths},
  journal = {Phys. Rev. Lett.},
  volume  = {127},
  number  = {25},
  pages   = {250404},
  year    = {2021},
  doi     = {10.1103/PhysRevLett.127.250404}
}

@article{Baczewski2013,
  author  = {Baczewski, Andrew D. and Bond, Stephen D.},
  title   = {Numerical Integration of the Extended Variable Generalized {Langevin} Equation with a Positive {Prony} Representable Memory Kernel},
  journal = {J. Chem. Phys.},
  volume  = {139},
  number  = {4},
  pages   = {044107},
  year    = {2013},
  doi     = {10.1063/1.4815917}
}

@article{Mascherpa2020,
  author  = {Mascherpa, F. and Smirne, A. and Somoza, A. D. and Fern{\'a}ndez-Acebal, P. and Donadi, S. and Tamascelli, D. and Huelga, S. F. and Plenio, M. B.},
  title   = {Optimized Auxiliary Oscillators for the Simulation of General Open Quantum Systems},
  journal = {Phys. Rev. A},
  volume  = {101},
  number  = {5},
  pages   = {052108},
  year    = {2020},
  doi     = {10.1103/PhysRevA.101.052108}
}

@article{Dorda2014,
  author  = {Dorda, Antonius and Nuss, Martin and von der Linden, Wolfgang and Arrigoni, Enrico},
  title   = {Auxiliary Master Equation Approach to Nonequilibrium Correlated Impurities},
  journal = {Phys. Rev. B},
  volume  = {89},
  number  = {16},
  pages   = {165105},
  year    = {2014},
  doi     = {10.1103/PhysRevB.89.165105}
}

@article{Dorda2015,
  author  = {Dorda, Antonius and Ganahl, Martin and Evertz, Hans Gerd and von der Linden, Wolfgang and Arrigoni, Enrico},
  title   = {Auxiliary Master Equation Approach within Matrix Product States: Spectral Properties of the Nonequilibrium Anderson Impurity Model},
  journal = {Phys. Rev. B},
  volume  = {92},
  number  = {12},
  pages   = {125145},
  year    = {2015},
  doi     = {10.1103/PhysRevB.92.125145}
}

@article{Lednev2024,
  author  = {Lednev, M. and Garc{\'i}a-Vidal, F. J. and Feist, J.},
  title   = {Lindblad Master Equation Capable of Describing Hybrid Quantum Systems in the Ultrastrong Coupling Regime},
  journal = {Phys. Rev. Lett.},
  volume  = {132},
  number  = {10},
  pages   = {106902},
  year    = {2024},
  doi     = {10.1103/PhysRevLett.132.106902}
}

@article{Alford2026,
  author  = {Alford, Wynter and Bettmann, Laetitia P. and Landi, Gabriel T.},
  title   = {Subtleties in the Pseudomodes Formalism},
  journal = {Phys. Rev. B},
  volume  = {113},
  number  = {20},
  pages   = {205143},
  year    = {2026},
  doi     = {10.1103/79dp-vz29}
}

@article{Huang2026,
  author  = {Huang, Zhen and Park, Gunhee and Chan, Garnet Kin-Lic and Lin, Lin},
  title   = {Coupled {Lindblad} Pseudomode Theory for Simulating Open Quantum Systems},
  journal = {Phys. Rev. Lett.},
  volume  = {136},
  number  = {9},
  pages   = {090403},
  year    = {2026},
  doi     = {10.1103/qfkp-jc7s}
}

@article{Park2024,
  author  = {Park, Gunhee and Huang, Zhen and Zhu, Yuanran and Yang, Chao and Chan, Garnet Kin-Lic and Lin, Lin},
  title   = {Quasi-{Lindblad} Pseudomode Theory for Open Quantum Systems},
  journal = {Phys. Rev. B},
  volume  = {110},
  number  = {19},
  pages   = {195148},
  year    = {2024},
  doi     = {10.1103/PhysRevB.110.195148}
}

@article{Thoenniss2025,
  author  = {Thoenniss, Julian and Vilkoviskiy, Ilya and Abanin, Dmitry A.},
  title   = {Efficient Pseudomode Representation and Complexity of Quantum Impurity Models},
  journal = {Phys. Rev. B},
  volume  = {112},
  number  = {15},
  pages   = {155114},
  year    = {2025},
  doi     = {10.1103/h8g7-bmng}
}

@article{VilkoviskiyAbanin2024,
  author  = {Vilkoviskiy, Ilya and Abanin, Dmitry A.},
  title   = {Bound on Approximating {Non-Markovian} Dynamics by Tensor Networks in the Time Domain},
  journal = {Phys. Rev. B},
  volume  = {109},
  number  = {20},
  pages   = {205126},
  year    = {2024},
  doi     = {10.1103/PhysRevB.109.205126}
}

@misc{HuangDecomp2026,
  author        = {Huang, Zhen and Ding, Zhiyan and Wang, Ke and Kaye, Jason and Li, Xiantao and Lin, Lin},
  title         = {Provably Efficient Long-Time Exponential Decompositions of {Non-Markovian} {Gaussian} Baths},
  year          = {2026},
  eprint        = {2603.25708}
}

@article{Machlup1954,
  author  = {Machlup, Stefan},
  title   = {Noise in Semiconductors: Spectrum of a Two-Parameter Random Signal},
  journal = {J. Appl. Phys.},
  volume  = {25},
  number  = {3},
  pages   = {341--343},
  year    = {1954},
  doi     = {10.1063/1.1721637}
}

@article{SauveSzabo1985,
  author  = {Sauv{\'e}, R. and Szabo, G.},
  title   = {Interpretation of {$1/f$} Fluctuations in Ion Conducting Membranes},
  journal = {J. Theor. Biol.},
  volume  = {113},
  number  = {3},
  pages   = {501--516},
  year    = {1985},
  doi     = {10.1016/S0022-5193(85)80035-7}
}

@article{Pellegrini1986,
  author  = {Pellegrini, B. and Neri, B. and Saletti, R.},
  title   = {Minimum Number of Lorentzian Spectra Sufficient to Yield a {$1/f^{\gamma}$} Spectrum},
  journal = {Alta Frequenza},
  volume  = {55},
  number  = {4},
  pages   = {245--253},
  year    = {1986}
}

@article{Milotti1995,
  author  = {Milotti, E.},
  title   = {Linear Processes That Produce {$1/f$} or Flicker Noise},
  journal = {Phys. Rev. E},
  volume  = {51},
  number  = {4},
  pages   = {3087--3103},
  year    = {1995},
  doi     = {10.1103/PhysRevE.51.3087}
}

@article{Yoshihara2006,
  author  = {Yoshihara, Fumiki and Harrabi, Khalil and Niskanen, Antti O. and Nakamura, Yasunobu and Tsai, Jaw-Shen},
  title   = {Decoherence of Flux Qubits due to {$1/f$} Flux Noise},
  journal = {Phys. Rev. Lett.},
  volume  = {97},
  number  = {16},
  pages   = {167001},
  year    = {2006},
  doi     = {10.1103/PhysRevLett.97.167001}
}

@article{Bylander2011,
  author  = {Bylander, Jonas and Gustavsson, Simon and Yan, Fei and Yoshihara, Fumiki and Harrabi, Khalil and Fitch, George and Cory, David G. and Nakamura, Yasunobu and Tsai, Jaw-Shen and Oliver, William D.},
  title   = {Noise Spectroscopy through Dynamical Decoupling with a Superconducting Flux Qubit},
  journal = {Nat. Phys.},
  volume  = {7},
  number  = {7},
  pages   = {565--570},
  year    = {2011},
  doi     = {10.1038/nphys1994}
}

@book{Akhiezer1992,
  author    = {Akhiezer, N. I.},
  title     = {Theory of Approximation},
  publisher = {Dover},
  address   = {New York},
  year      = {1992}
}

@article{IstaceThiran1995,
  author  = {Istace, M.-P. and Thiran, J.-P.},
  title   = {On the Third and Fourth {Zolotarev} Problems in the Complex Plane},
  journal = {SIAM J. Numer. Anal.},
  volume  = {32},
  number  = {1},
  pages   = {249--259},
  year    = {1995},
  doi     = {10.1137/0732009}
}

@article{Gawlik2019,
  author  = {Gawlik, Evan S.},
  title   = {{Zolotarev} Iterations for the Matrix Square Root},
  journal = {SIAM J. Matrix Anal. Appl.},
  volume  = {40},
  number  = {2},
  pages   = {696--719},
  year    = {2019},
  doi     = {10.1137/18M1178529}
}

@article{TrefethenWilber2025,
  author  = {Trefethen, Lloyd N. and Wilber, Heather D.},
  title   = {Computation of {Zolotarev} Rational Functions},
  journal = {SIAM J. Sci. Comput.},
  volume  = {47},
  number  = {4},
  pages   = {A2205--A2220},
  year    = {2025},
  doi     = {10.1137/24M1687960}
}

@article{Lax1963,
  author  = {Lax, Melvin},
  title   = {Formal Theory of Quantum Fluctuations from a Driven State},
  journal = {Phys. Rev.},
  volume  = {129},
  number  = {5},
  pages   = {2342--2348},
  year    = {1963},
  doi     = {10.1103/PhysRev.129.2342}
}

@article{Cywinski2008,
  author  = {Cywi{\'n}ski, {\L}ukasz and Lutchyn, Roman M. and Nave, Cody P. and Das Sarma, S.},
  title   = {How to Enhance Dephasing Time in Superconducting Qubits},
  journal = {Phys. Rev. B},
  volume  = {77},
  number  = {17},
  pages   = {174509},
  year    = {2008},
  doi     = {10.1103/PhysRevB.77.174509}
}

@article{AlvarezSuter2011,
  author  = {{\'A}lvarez, Gonzalo A. and Suter, Dieter},
  title   = {Measuring the Spectrum of Colored Noise by Dynamical Decoupling},
  journal = {Phys. Rev. Lett.},
  volume  = {107},
  number  = {23},
  pages   = {230501},
  year    = {2011},
  doi     = {10.1103/PhysRevLett.107.230501}
}

@article{Cerfontaine2021,
  author  = {Cerfontaine, P. and Hangleiter, T. and Bluhm, H.},
  title   = {Filter Functions for Quantum Processes under Correlated Noise},
  journal = {Phys. Rev. Lett.},
  volume  = {127},
  number  = {17},
  pages   = {170403},
  year    = {2021},
  doi     = {10.1103/PhysRevLett.127.170403}
}

@article{Vieira2024,
  author  = {Vieira, Lucas B. and Milz, Simon and Vitagliano, Giuseppe and Budroni, Costantino},
  title   = {Witnessing Environment Dimension through Temporal Correlations},
  journal = {Quantum},
  volume  = {8},
  pages   = {1224},
  year    = {2024},
  doi     = {10.22331/q-2024-01-10-1224}
}

@article{StornPrice1997,
  author  = {Storn, Rainer and Price, Kenneth},
  title   = {Differential Evolution---A Simple and Efficient Heuristic for Global Optimization over Continuous Spaces},
  journal = {J. Global Optim.},
  volume  = {11},
  number  = {4},
  pages   = {341--359},
  year    = {1997},
  doi     = {10.1023/A:1008202821328}
}

@article{NelderMead1965,
  author  = {Nelder, J. A. and Mead, R.},
  title   = {A Simplex Method for Function Minimization},
  journal = {Comput. J.},
  volume  = {7},
  number  = {4},
  pages   = {308--313},
  year    = {1965},
  doi     = {10.1093/comjnl/7.4.308}
}
\onecolumngrid

\clearpage
\onecolumngrid
\setcounter{secnumdepth}{1}
\setcounter{section}{0}
\renewcommand{\thesection}{S\arabic{section}}
\renewcommand{\theHsection}{S\arabic{section}}
\setcounter{equation}{0}
\renewcommand{\theequation}{S\arabic{equation}}
\renewcommand{\theHequation}{S\arabic{equation}}
\setcounter{figure}{0}
\renewcommand{\thefigure}{S\arabic{figure}}
\renewcommand{\theHfigure}{S\arabic{figure}}
\setcounter{table}{0}
\renewcommand{\thetable}{S\arabic{table}}
\renewcommand{\theHtable}{S\arabic{table}}
\setcounter{proposition}{0}
\renewcommand{\theproposition}{S\arabic{proposition}}
\renewcommand{\theHproposition}{S\arabic{proposition}}
\setcounter{lemma}{0}
\renewcommand{\thelemma}{S\arabic{lemma}}
\renewcommand{\theHlemma}{S\arabic{lemma}}
\setcounter{corollary}{0}
\renewcommand{\thecorollary}{S\arabic{corollary}}
\renewcommand{\theHcorollary}{S\arabic{corollary}}

\begin{center}
{\Large\bfseries Supplemental Material}\\[0.35em]
{\large\bfseries \paperTitle}
\end{center}

\section{Physical model, spectral positivity, and restoration of physical units}\label{sec:SM-positivity}

The proof in this Supplemental Material proceeds through spectral positivity and the rational-degree bound, the exact isometry between sign-function and inverse-square-root approximations, physical realization with positive residues, and inversion of the exact error relation. We then present spectral probing with nonnegative weights, lower bounds on the auxiliary-mode count from spectral data, and numerical cross-checks.

We collect the annihilation operators of the $n$ auxiliary modes into the vector $\hat{\bm b}=(\hat b_1,\ldots,\hat b_n)^T$. Here, passivity means that the linear drift closes on $\hat{\bm b}$ and does not mix in $\hat{\bm b}^{\dagger}$. The number-conserving quadratic Hamiltonian and the Lindblad dissipator generated by linear annihilation operators are
\begin{align}
\hat H_A&=\hat{\bm b}^\dagger H\hat{\bm b},
\qquad H=H^\dagger,\label{eq:SM-H}\\
\mathcal D_A(\rho)&=
\sum_{j,k=1}^{n}2\Gamma_{jk}
\left(
\hat b_k\rho\hat b_j^\dagger
-\tfrac12\{\hat b_j^\dagger\hat b_k,\rho\}
\right),
\qquad \Gamma\succeq0 .
\label{eq:SM-D}
\end{align}
The system--auxiliary-network interaction is written as
$\hat H_{SA}=\hat S\otimes\hat B$, where $\hat S=\hat S^\dagger$, and the Hermitian bath operator seen by the system is
\begin{equation}
\hat B=\bm g^\dagger\hat{\bm b}
+\hat{\bm b}^\dagger\bm g\, .
\label{eq:SM-B}
\end{equation}
Here, $\bm g$ is the system--auxiliary-mode coupling vector. Because the system probes the network through only this single Hermitian bath operator, the spectrum seen by the system is scalar. This restriction applies only to the coupling interface through which the system probes the network; the internal Markovian dissipation of the auxiliary network is still described by an arbitrary $\Gamma\succeq0$. With the convention for the Lindblad coefficients used here, the homogeneous evolution of the first moments of the auxiliary modes obeys
$\dd\langle\hat{\bm b}\rangle/\dd t=-(\Gamma+iH)\langle\hat{\bm b}\rangle$. We therefore define the drift matrix $M=\Gamma+iH$. Here, $I_n$ denotes the identity matrix of order $n$, and $\sigma(M)$ denotes the set of eigenvalues of $M$. The stability condition is
$\min_{\lambda\in\sigma(M)}\Ree\lambda>0$; it ensures that $e^{-Mt}$ decays and excludes poles on the real-frequency axis.

\begin{proposition}[Spectral positivity and high-frequency decay]
\label{prop:positivity}
When the auxiliary network is in its stationary vacuum state, the spectrum in Eq.~\eqref{eq:model-spectrum} satisfies
\begin{equation}
J(\omega)\ge0,
\qquad
J(\omega)=\frac{2\bm g^\dagger\Gamma\bm g}{\omega^2}
+O(\omega^{-3}) .
\label{eq:SM-positive-asymptotic}
\end{equation}
\end{proposition}

\begin{proof}
The quantum regression theorem shows that the bath correlation function
$\langle\hat B(t)\hat B(0)\rangle_0$
follows the same homogeneous evolution as the single-time expectation values. Thus, in the stationary vacuum state,
$C(t)=\langle\hat B(t)\hat B(0)\rangle_0
=\bm g^\dagger e^{-Mt}\bm g$ for $t\ge0$~\cite{Lax1963}.
Stationarity and $\hat B=\hat B^\dagger$ give $C(-t)=C(t)^*$, so the two-sided spectrum is
\begin{equation}
J(\omega)=\int_{-\infty}^{\infty}C(t)e^{i\omega t}\dd t
=2\Ree\!\left[\bm g^\dagger(M-i\omega I_n)^{-1}\bm g\right].
\label{eq:SM-Fourier}
\end{equation}
Let $\bm y_\omega=(M-i\omega I_n)^{-1}\bm g$. Since
$M+M^\dagger=2\Gamma$,
\begin{equation}
J(\omega)=
\bm y_\omega^\dagger(M+M^\dagger)\bm y_\omega
=2\bm y_\omega^\dagger\Gamma\bm y_\omega\ge0 .
\label{eq:SM-positive-form}
\end{equation}
On the other hand,
\begin{equation}
(M-i\omega I_n)^{-1}
=\frac{iI_n}{\omega}+\frac{M}{\omega^2}+O(\omega^{-3}),
\end{equation}
and taking the real part yields Eq.~\eqref{eq:SM-positive-asymptotic}.
\end{proof}

If the original physical target is
$J_{\rm tar}^{\rm phys}(\nu)=\mathcal A_0/|\nu|$
on $\nu_{\min}\le|\nu|\le\nu_{\max}$, then
\begin{equation}
\omega=\frac{\nu}{\nu_{\min}},
\qquad
R=\frac{\nu_{\max}}{\nu_{\min}},
\qquad
M^{\rm phys}=\nu_{\min}M,
\qquad
\bm g^{\rm phys}=\sqrt{\mathcal A_0}\,\bm g .
\label{eq:SM-units}
\end{equation}
For independent auxiliary modes, physical units are restored according to $\gamma_j^{\rm phys}=\nu_{\min}\gamma_j$,
$|g_j^{\rm phys}|^2=\mathcal A_0|g_j|^2$, and the partial-fraction residues scale as
$a_j^{\rm phys}=\mathcal A_0\nu_{\min}a_j$.

\section{Rational-degree bound}\label{sec:SM-degree}

Define
\begin{equation}
\Rrat_{a,b}^{\mathbb R}(I)=
\left\{
\frac{u}{v}:\ 
\begin{gathered}
u,v\in\mathbb R[x],\ 
\deg u\le a,\ 
\deg v\le b,\\
v(x)\ne0,\quad \forall x\in I
\end{gathered}
\right\} .
\label{eq:SM-rational-class}
\end{equation}
We assume throughout that common factors in all rational functions have been canceled. For a polynomial or rational function with complex coefficients, define
\begin{equation}
w^\#(\lambda)=\overline{w(\bar\lambda)} .
\label{eq:SM-sharp}
\end{equation}

\begin{lemma}[Rational-degree bound]
\label{lem:degree}
Every stable $n$-mode network satisfies
\begin{equation}
J\in\Rrat_{2n-2,2n}^{\mathbb R}(\Omega_R),
\qquad
\omega J(\omega)\in\Rrat_{2n-1,2n}^{\mathbb R}(\Omega_R) .
\label{eq:SM-degree}
\end{equation}
This result does not require $M$ to be normal or diagonalizable. It therefore applies both to networks with nonnormal structure and to those containing Jordan blocks that generate higher-order poles.
\end{lemma}

\begin{proof}
Let
\begin{equation}
v(\lambda)=\det(M-i\lambda I_n),
\qquad
u(\lambda)=\bm g^\dagger\operatorname{adj}(M-i\lambda I_n)\bm g ,
\end{equation}
where $\operatorname{adj}(A)$ denotes the adjugate matrix of $A$. Then
\begin{equation}
h(\lambda)=\bm g^\dagger(M-i\lambda I_n)^{-1}\bm g
=\frac{u(\lambda)}{v(\lambda)},
\qquad
\deg u\le n-1,
\quad \deg v=n .
\label{eq:SM-resolvent-ratio}
\end{equation}
This derivation does not require $M$ to be diagonalizable. Stability ensures that $v(\omega)\ne0$ for every real $\omega$. On the real axis,
\begin{equation}
J(\omega)=
\frac{u(\omega)v^\#(\omega)+u^\#(\omega)v(\omega)}
{v(\omega)v^\#(\omega)} .
\label{eq:SM-sharp-spectrum}
\end{equation}
Both the numerator and denominator are invariant under $\#$ conjugation and therefore have real coefficients. Before cancellation, the denominator has degree at most $2n$. The decay $J(\omega)=O(\omega^{-2})$ established in Proposition~\ref{prop:positivity} forces the potentially present $\omega^{2n-1}$ term in the numerator to vanish, so the numerator has degree at most $2n-2$. Multiplication by $\omega$ then gives Eq.~\eqref{eq:SM-degree}.
\end{proof}

If a frequency-independent white-noise feedthrough term $d>0$ is included, the numerator degree of
$q(\omega)=\omega[d+J(\omega)]$ can be as large as $2n+1$.
The result then lies outside the degree class in Eq.~\eqref{eq:SM-degree}; throughout this work, we take $d=0$.

For any $q(\omega)=\omega J(\omega)$ generated by the networks above, define
$\widetilde q(\xi)=q(R\xi)$. By Lemma~\ref{lem:degree} and Eq.~\eqref{eq:sign-lift}, this scaling maps $\Omega_R$ to
$I_{R^{-1}}$, leaves the rational degrees unchanged, and satisfies
\begin{equation*}
\sup_{\omega\in\Omega_R}
|\sgnfun(\omega)-q(\omega)|
=
\sup_{\xi\in I_{R^{-1}}}
|\sgnfun(\xi)-\widetilde q(\xi)| .
\end{equation*}
All such $\widetilde q$ belong to
$\Rrat_{2N-1,2N}^{\mathbb R}(I_{R^{-1}})$, which is contained in the full
$\Rrat_{2N,2N}^{\mathbb R}(I_{R^{-1}})$ class used in the fourth Zolotarev problem. Taking the infimum over a larger admissible class can only yield a smaller or equal error; hence
\begin{equation}
\Ephys_N^{\rm coup}(R)\ge e_{2N}(R^{-1}) .
\label{eq:SM-lower-bound}
\end{equation}
Any auxiliary mode that, after cancellation, does not appear in the frequency response seen by the system only increases the state-space dimension; it cannot increase the numerator or denominator degree of $q(\omega)$ after cancellation and therefore does not affect the lower bound in Eq.~\eqref{eq:SM-lower-bound}.

\section{Exact isometry between sign-function and inverse-square-root approximations}\label{sec:SM-isometry}

\begin{lemma}[Isometry between sign-function and inverse-square-root approximations]
\label{lem:isometry}
For an integer $N\ge1$, $0<k<1$, and any
$r\in\Rrat_{N-1,N}^{\mathbb R}([k^2,1])$, let
$q_r(\xi)=\xi r(\xi^2)$. Then
\begin{equation}
\sup_{\xi\in[-1,-k]\cup[k,1]}
\left|\sgnfun(\xi)-q_r(\xi)\right|
=
\sup_{x\in[k^2,1]}
\left|1-\sqrt{x}\,r(x)\right| .
\label{eq:SM-sign-sqrt-isometry}
\end{equation}
\end{lemma}

\begin{proof}
Since $\xi=\sgnfun(\xi)|\xi|$, pointwise
\begin{equation}
\left|\sgnfun(\xi)-\xi r(\xi^2)\right|
=
\left|1-|\xi|r(\xi^2)\right| .
\end{equation}
Setting $x=\xi^2$ maps both bands to $[k^2,1]$ and proves the result.
\end{proof}

\begin{lemma}[Even-degree Zolotarev type and error correspondence]
\label{lem:zolo-type}
For an integer $N\ge1$ and $0<k<1$, on the standard two-interval set
$I_k=[-1,-k]\cup[k,1]$,
the optimal error over the full class
$\Rrat_{2N,2N}^{\mathbb R}(I_k)$
can be attained by an odd function that, after cancellation, is of type
$(2N-1,2N)$. Moreover,
\begin{equation}
e_{2N}(k)=
\min_{r\in\Rrat_{N-1,N}^{\mathbb R}([k^2,1])}
\max_{x\in[k^2,1]}
\left|1-\sqrt{x}\,r(x)\right| .
\label{eq:SM-zolo-correspondence}
\end{equation}
\end{lemma}

\begin{proof}
We need only two facts from classical Zolotarev theory: the optimal error over the full type-$(2N,2N)$ class can be attained by an odd function that, after cancellation, is of type $(2N-1,2N)$; after factoring out $z$, the remaining rational function, viewed as a function of $z^2$, belongs precisely to the type-$(N-1,N)$ class needed here for inverse-square-root approximation.
Section 3.1 of the arXiv v1 version of Ref.~\cite{Gawlik2019}, immediately before its Eq.~(22), gives the corresponding odd function
$s_{2\ell+1,2m}(z;k)$
in the explicit product form
\begin{equation}
\begin{aligned}
s_{2\ell+1,2m}(z;k)
&=\mathscr M_{m,\ell}(k)z\,
\frac{\prod_{j=1}^{\ell}(z^2+c_{2j}(k))}
{\prod_{j=1}^{m}(z^2+c_{2j-1}(k))},\\
c_j(k)
&=k^2
\frac{\operatorname{sn}^2\!\left(jK(k')/(m+\ell+1);k'\right)}
{\operatorname{cn}^2\!\left(jK(k')/(m+\ell+1);k'\right)},
\qquad
k'=\sqrt{1-k^2}.
\end{aligned}
\label{eq:SM-Gawlik-factor}
\end{equation}
Here $\ell\in\{m-1,m\}$ and $j=1,\ldots,m+\ell$; an empty product in the numerator is defined to equal $1$. The constant $\mathscr M_{m,\ell}(k)>0$ is the unique normalization
constant for which the maximum positive deviation of
$s_{2\ell+1,2m}(z;k)-1$ equals the absolute value of its minimum
deviation on $z\in[k,1]$. The coefficients $c_j(k)$ in Eq.~\eqref{eq:SM-Gawlik-factor} are defined in Eq.~(22) of that arXiv version in terms of Jacobi elliptic functions; Eq.~(23) of the same version gives the alternation points. Setting $m=N$ and $\ell=N-1$, the function
$s_{2N-1,2N}(z;k)$
is, after cancellation, of type $(2N-1,2N)$; Eq.~(27) of that version further shows that it attains the optimal error over the full type-$(2N,2N)$ class. After the factor $z$ is extracted from Eq.~\eqref{eq:SM-Gawlik-factor}, the remaining rational function, as a function of $z^2$, has numerator and denominator degrees $N-1$ and $N$, respectively. Combining this observation with Lemma~\ref{lem:isometry} gives the required inverse-square-root form and establishes Eq.~\eqref{eq:SM-zolo-correspondence}.
\end{proof}

To avoid ambiguity between the standard and physical intervals, for a given $0<k<1$ let
$s_{2N}^{\rm std}(\xi;k)$ be an optimal rational approximant on the standard two-interval set
$I_k=[-1,-k]\cup[k,1]$ for the degree-$2N$ Zolotarev problem for sign approximation, satisfying $\sup_{\xi\in I_k}|\sgnfun(\xi)-s_{2N}^{\rm std}(\xi;k)|=e_{2N}(k)$.
Let
$r_N^{\rm std}(x;k)
\in\Rrat_{N-1,N}^{\mathbb R}([k^2,1])$
be the optimal rational approximant corresponding to $s_{2N}^{\rm std}(\xi;k)$ for the inverse-square-root minimax problem in relative error, with the normalization
\begin{equation}
s_{2N}^{\rm std}(\xi;k)
=\xi\,r_N^{\rm std}(\xi^2;k),
\qquad
\max_{x\in[k^2,1]}
|1-\sqrt{x}\,r_N^{\rm std}(x;k)|
=e_{2N}(k).
\label{eq:SM-standard-extremals}
\end{equation}
Set
\begin{equation}
k=R^{-1},
\qquad
\xi=\frac{\omega}{R},
\qquad
x_{\rm std}=\xi^2=\frac{x}{R^2},
\qquad x=\omega^2 .
\label{eq:SM-scaling}
\end{equation}
After scaling to the physical interval according to Eq.~\eqref{eq:SM-scaling}, we obtain
\begin{align}
q_{N,R}^{\rm opt}(\omega)
&=s_{2N}^{\rm std}(\omega/R;R^{-1})\notag\\
&=\omega r_{N,R}^{\rm opt}(\omega^2),
\label{eq:SM-scale-q}\\
r_{N,R}^{\rm opt}(x)
&=R^{-1}r_N^{\rm std}(x/R^2;R^{-1}) .
\label{eq:SM-scale-r}
\end{align}
Because $x=\omega^2$, both the positive- and negative-frequency bands map to
$[1,R^2]$. Equations~\eqref{eq:SM-scale-q} and
\eqref{eq:SM-scale-r} give the positive-interval form in Eq.~\eqref{eq:main-theorem} of the main text. In the next section, we further prove that the poles of
$r_{N,R}^{\rm opt}(x)$
lie on the negative real axis and that its residues are positive. Its partial-fraction expansion can therefore be realized directly by independent damped auxiliary modes at zero detuning; in the spectrum, each auxiliary mode contributes one positive Lorentzian component.

\section{Positive-residue construction and proof of the main theorem}\label{sec:SM-realization}

Let
\begin{equation}
k=R^{-1},
\qquad
k'=\sqrt{1-k^2},
\qquad
K'=K(k') .
\end{equation}
Set $m=N$ and $\ell=N-1$. After the rescaling in Eq.~\eqref{eq:SM-scaling}, the corresponding coefficients on the physical interval $x\in[1,R^2]$ are
\begin{equation}
c_j(N,R)=R^2c_j(R^{-1})=
\frac{\operatorname{sn}^2\!\left(j K'/(2N);k'\right)}
{1-\operatorname{sn}^2\!\left(j K'/(2N);k'\right)},
\qquad j=1,\ldots,2N-1 .
\label{eq:SM-jacobi-coefficients}
\end{equation}
We henceforth abbreviate $c_j\equiv c_j(N,R)$. Here,
$\operatorname{sn}(u;k')$
is a Jacobi elliptic function, and we adopt the modulus convention. These coefficients satisfy
\begin{equation}
0<c_1<c_2<\cdots<c_{2N-1} .
\label{eq:SM-coefficient-order}
\end{equation}
The odd-indexed coefficients $c_{2j-1}$ enter the denominator and give the poles
$-c_{2j-1}$, whereas the even-indexed coefficients $c_{2j}$ enter the numerator and give the zeros
$-c_{2j}$. This ordering gives strict pole--zero interlacing along the negative real axis.
For numerical implementations, software that takes the elliptic parameter rather than the modulus as input should be supplied with
$k'^2$.
Define
\begin{equation}
P_{N,R}(x)=
\frac{\prod_{j=1}^{N-1}(x+c_{2j})}
{\prod_{j=1}^{N}(x+c_{2j-1})},
\qquad
F_{N,R}(x)=\sqrt{x}\,P_{N,R}(x) .
\label{eq:SM-product}
\end{equation}
Let
\begin{equation}
F_-=
\min_{1\le x\le R^2}F_{N,R}(x),
\qquad
F_+=
\max_{1\le x\le R^2}F_{N,R}(x),
\end{equation}
and define the normalization factor and optimal error by
\begin{equation}
\chi_{N,R}=\frac{2}{F_-+F_+},
\qquad
\Ephys_N(R)=\frac{F_+-F_-}{F_++F_-} .
\label{eq:SM-normalization}
\end{equation}
Because $P_{N,R}(x)>0$ for every $x\in[1,R^2]$, we have
$F_->0$, $F_+>0$, and hence $\chi_{N,R}>0$.
Equation~\eqref{eq:SM-normalization} gives
$\chi_{N,R}F_-=1-\Ephys_N(R)$ and
$\chi_{N,R}F_+=1+\Ephys_N(R)$, so the maximum positive deviation equals
the absolute value of the minimum deviation.

The rescaling in Eq.~\eqref{eq:SM-scaling} also gives
$\chi_{N,R}=R\,\mathscr M_{N,N-1}(R^{-1})$, showing that
$\chi_{N,R}$ is the counterpart, in the physical variable, of the normalization constant on the standard interval. Specifically, substituting
$x_{\rm std}=x/R^2$
and
$c_j(R^{-1})=c_j(N,R)/R^2$
into the standard product shows that the ratio of the numerator scale factor to the denominator scale factor is $R^2$. Combining this ratio with the overall factor $R^{-1}$ in
Eq.~\eqref{eq:SM-scale-r}
yields the stated relation.

Applying $x=z^2$ and the rescaling in Eq.~\eqref{eq:SM-scaling} to the positive-interval alternation points given by Eq.~(23) of the arXiv v1 version of Ref.~\cite{Gawlik2019} yields $2N+1$ alternation points at which
$1-\sqrt{x}\,\chi_{N,R}P_{N,R}(x)$ equioscillates. For a type-$(N-1,N)$ rational approximation, these $2N+1$ extrema alternate in sign and have equal magnitude, satisfying the Chebyshev alternation criterion for rational approximation. Therefore,
$\chi_{N,R}P_{N,R}(x)$
is minimax optimal among rational functions of this type~\cite{Akhiezer1992,IstaceThiran1995}.

\begin{lemma}[Positive residues]
\label{lem:positive-residues}
The function
\begin{equation}
r_{N,R}^{\rm opt}(x)=\chi_{N,R}P_{N,R}(x)
\end{equation}
has the partial-fraction expansion
\begin{equation}
r_{N,R}^{\rm opt}(x)=
\sum_{j=1}^{N}\frac{a_j}{x+c_{2j-1}},
\qquad a_j>0 .
\label{eq:SM-partial-fraction}
\end{equation}
\end{lemma}

\begin{proof}
By Eqs.~\eqref{eq:SM-product} and
\eqref{eq:SM-coefficient-order},
the denominator has $N$ distinct simple zeros
$-c_{2j-1}$. They strictly interlace with the numerator zeros, so no common factors cancel and the denominator zeros are simple poles.
Since the numerator and denominator degrees are $N-1$ and $N$, respectively, the partial-fraction expansion has no polynomial part, and its coefficients are the residues at the simple poles:
\begin{equation}
a_j=\chi_{N,R}
\frac{\prod_{u=1}^{N-1}(c_{2u}-c_{2j-1})}
{\prod_{\substack{v=1\\v\ne j}}^{N}(c_{2v-1}-c_{2j-1})} .
\label{eq:SM-residues}
\end{equation}
By Eq.~\eqref{eq:SM-coefficient-order}, the factors with $u<j$ in the numerator are negative, and there are $j-1$ of them; the factors with $v<j$ in the denominator are likewise negative, and there are also $j-1$ of them. Their signs cancel, and $\chi_{N,R}>0$; hence $a_j>0$.
\end{proof}

Set
\begin{equation}
\gamma_j=\sqrt{c_{2j-1}},
\qquad
|g_j|=\sqrt{\frac{a_j}{2\gamma_j}},
\qquad
H_{\rm opt}=0,
\qquad
\Gamma_{\rm opt}=\operatorname{diag}(\gamma_1,\ldots,\gamma_N) .
\label{eq:SM-physical-parameters}
\end{equation}
The corresponding spectrum is
\begin{equation}
J_{N,R}^{\rm opt}(\omega)
:=r_{N,R}^{\rm opt}(\omega^2)
=2\sum_{j=1}^{N}
\frac{|g_j|^2\gamma_j}{\omega^2+\gamma_j^2}.
\label{eq:SM-physical-spectrum}
\end{equation}
This construction is stable and passive and has a positive-semidefinite damping matrix; it therefore generates a completely positive and trace-preserving Lindblad evolution.

\begin{proof}[Proof of Theorem~\ref{thm:main}]
Lemma~\ref{lem:degree} gives
$\Ephys_N^{\rm coup}(R)\ge e_{2N}(R^{-1})$. Lemmas~\ref{lem:zolo-type} and~\ref{lem:positive-residues} provide an explicit physical construction in the uncoupled diagonal subclass whose error is exactly $e_{2N}(R^{-1})$; hence
$\Ephys_N^{\rm diag}(R)\le e_{2N}(R^{-1})$. Since
$\Pclass_{\le N}^{\rm diag}\subseteq\Pclass_{\le N}^{\rm coup}$, we also have
$\Ephys_N^{\rm coup}(R)\le\Ephys_N^{\rm diag}(R)$. Therefore,
$e_{2N}(R^{-1})\le\Ephys_N^{\rm coup}(R)\le\Ephys_N^{\rm diag}(R)\le e_{2N}(R^{-1})$, which proves the theorem.
\end{proof}

For $N=1$, the numerator in Eq.~\eqref{eq:SM-product} is an empty
product. Equation~\eqref{eq:SM-jacobi-coefficients} and the identity
$\operatorname{sn}^2(K'/2;k')=1/(1+k)$ give $c_1=R$. Hence
$P_{1,R}(x)=1/(x+R)$ and
$F_{1,R}(x)=\sqrt{x}/(x+R)$. The minimum of $F_{1,R}$ is
$F_-=1/(1+R)$ at the two endpoints, whereas its maximum is
$F_+=1/(2\sqrt R)$ at $x=R$.
Equations~\eqref{eq:SM-normalization} and
\eqref{eq:SM-physical-parameters} then give
$\gamma_1=\sqrt R$ and
$\Ephys_1(R)=[(\sqrt R-1)/(\sqrt R+1)]^2$.
As a concrete numerical example, Table~\ref{tab:R100N3} lists the
dimensionless optimal parameters for $R=100$ and $N=3$; the
corresponding error is $0.0285595879\ldots$.

\begin{table}[!ht]
\caption{Optimal parameters of the independent damped auxiliary modes at zero detuning for $R=100$ and $N=3$. Here, $a_j$ is a dimensionless residue; physical units are restored according to Eq.~\eqref{eq:SM-units} and the accompanying text.}
\label{tab:R100N3}
\begin{ruledtabular}
\begin{tabular}{cccc}
$j$ & $\gamma_j$ & $a_j$ & $|g_j|$\\
\hline
1 & $1.17307049$ & $1.95962983$ & $0.91392376$\\
2 & $10.00000000$ & $12.83907370$ & $0.80122012$\\
3 & $85.24636904$ & $142.40519095$ & $0.91392376$
\end{tabular}
\end{ruledtabular}
\end{table}

\section{Inverting the exact error relation for mode count and positive-frequency dynamic range}\label{sec:SM-resources}

To invert the exact error relation in the main text, we retain the function $\mu(k)$ defined in Eq.~\eqref{eq:mu-main}. This function is continuous and strictly decreasing on $(0,1)$ and satisfies
\begin{equation}
\mu(k)\to\infty\quad(k\downarrow0),
\qquad
\mu(k)\to0\quad(k\uparrow1) .
\label{eq:SM-mu-limits}
\end{equation}
The following Zolotarev modular equations therefore have unique solutions. The standard modular equation and error parametrization for the fourth Zolotarev problem are given in Refs.~\cite{Akhiezer1992,IstaceThiran1995,Gawlik2019}; in our notation, they yield the following corollary.

\begin{corollary}[Exact error and mode-count inversion]
\label{cor:SM-resource}
For any $R>1$ and integer $N\ge1$, let
$\zeta_{N,R}\in(0,1)$ be the unique solution of
\begin{equation}
2N\mu(\zeta_{N,R})=\mu(R^{-1}) .
\label{eq:SM-modular}
\end{equation}
Then
\begin{equation}
\Ephys_N(R)=\frac{1-\zeta_{N,R}}{1+\zeta_{N,R}} .
\label{eq:SM-exact-error}
\end{equation}
For $0<\varepsilon<1$, the condition
$\Ephys_N(R)\le\varepsilon$
is equivalent to $\zeta_{N,R}\ge\frac{1-\varepsilon}{1+\varepsilon}$. Because $\mu$ is strictly decreasing, this condition is further equivalent to $\mu(\zeta_{N,R})\le\mu\!\left(\frac{1-\varepsilon}{1+\varepsilon}\right)$. Together with Eq.~\eqref{eq:SM-modular}, this gives $N\ge\frac{\mu(R^{-1})}{2\mu((1-\varepsilon)/(1+\varepsilon))}$.
Taking the smallest integer $N$ that satisfies this condition yields Eq.~\eqref{eq:Nmin-main} of the main text.
If $\varepsilon=0$, no spectrum generated by a finite number of auxiliary modes can coincide exactly with the target spectrum $1/|\omega|$ over the entire two-sided frequency band; if $\varepsilon\ge1$, the zero spectrum is already feasible.
\end{corollary}

For fixed $0<\varepsilon<1$ and integer $N\ge1$, define
\begin{equation}
\tau_\varepsilon=\frac{1-\varepsilon}{1+\varepsilon} .
\end{equation}
If $k_{N,\varepsilon}\in(0,1)$ is the unique solution of
\begin{equation}
\mu(k_{N,\varepsilon})=2N\mu(\tau_\varepsilon),
\label{eq:SM-bandwidth-inverse}
\end{equation}
then
\begin{equation}
R_{\max}(N,\varepsilon)=k_{N,\varepsilon}^{-1},
\qquad
D_{\max}(N,\varepsilon)
=-\log_{10}k_{N,\varepsilon} .
\label{eq:SM-Dmax}
\end{equation}
The asymptotic formulas for the mode count and positive-frequency dynamic range follow from the standard expansions of the complete elliptic integrals~\cite{Akhiezer1992}. As $k\downarrow0$,
\[
K(k)=\frac{\pi}{2}+O(k^2),
\qquad
K(\sqrt{1-k^2})=\ln\frac{4}{k}
+O\!\left(k^2\ln\frac{1}{k}\right),
\]
and hence $\mu(k)=\ln\frac{4}{k}+O\!\left(k^2\ln\frac{1}{k}\right)$.
On the other hand, let
$\tau_\varepsilon'=\sqrt{1-\tau_\varepsilon^2}
=2\sqrt{\varepsilon}/(1+\varepsilon)$. As $\varepsilon\downarrow0$,
\[
K(\tau_\varepsilon')=\frac{\pi}{2}+O(\varepsilon),
\qquad
K(\tau_\varepsilon)=\frac12\ln\frac{4}{\varepsilon}
+O\!\left(\varepsilon\ln\frac{1}{\varepsilon}\right),
\]
so that $\mu(\tau_\varepsilon)=\frac{\pi^2}{2\ln(4/\varepsilon)}[1+o(1)]$.
Substituting these two expansions into Eq.~\eqref{eq:Nmin-main} of the main text, and noting that the ceiling contributes only an $O(1)$ correction, gives Eq.~\eqref{eq:Nmin-asymptotic} of the main text.

For fixed $0<\varepsilon<1$, $2N\mu(\tau_\varepsilon)$ grows linearly with $N$, so $k_{N,\varepsilon}\downarrow0$ as $N\to\infty$. The small-$k$ expansion above gives $k_{N,\varepsilon}=4e^{-2N\mu(\tau_\varepsilon)}[1+o(1)]$, and therefore
\begin{equation}
D_{\max}(N,\varepsilon)=
\frac{2\mu(\tau_\varepsilon)}{\ln10}N
-\log_{10}4+o(1) .
\label{eq:SM-Dmax-asymptotic}
\end{equation}
Taking $\varepsilon\downarrow0$ in turn, the expansion of $\mu(\tau_\varepsilon)$ above gives
\begin{equation}
\frac{2\mu(\tau_\varepsilon)}{\ln10}
=\frac{\pi^2}{\ln10\,\ln(4/\varepsilon)}[1+o(1)] .
\label{eq:SM-slope-small-error}
\end{equation}
Equation~\eqref{eq:SM-Dmax-asymptotic} is obtained by taking $N\to\infty$ at fixed $\varepsilon$, whereas Eq.~\eqref{eq:SM-slope-small-error} further gives the slope at small tolerance. Equation~\eqref{eq:Nmin-asymptotic} of the main text considers the joint asymptotic regime $R\to\infty$ and $\varepsilon\downarrow0$. These results are consistent but correspond to different limiting procedures.

\section{Spectral probing with nonnegative frequency weights and mode-count lower bounds from spectral data}\label{sec:SM-data}

\subsection{Dual characterization in terms of nonnegative frequency weights}

Filter functions in quantum decoherence and noise spectroscopy are physical examples of continuous nonnegative frequency weights~\cite{Cywinski2008,AlvarezSuter2011,Cerfontaine2021}. The following proposition gives the supremum, over all such weights, of the relative deviation of a weighted spectral quantity.

\begin{proposition}[Dual characterization in terms of nonnegative frequency weights]
\label{prop:positive-weight-duality}
Let $\mathcal C(\Omega_R)$ denote the space of continuous functions on $\Omega_R$. Suppose that $J_{\rm tar}\in\mathcal C(\Omega_R)$, $J_{\rm tar}>0$, and $J/J_{\rm tar}\in\mathcal C(\Omega_R)$. Define the class of nonnegative weights $\Wclass_R$ and the weighted spectral functional $\Phi_W[J]$ by
\begin{equation}
\Wclass_R=
\{W\in\mathcal C(\Omega_R):W\ge0,\ W\not\equiv0\},
\qquad
\Phi_W[J]=\frac{1}{2\pi}
\int_{\Omega_R}W(\omega)J(\omega)\dd\omega .
\end{equation}
Then
\begin{equation}
\sup_{W\in\Wclass_R}
\left|
\frac{\Phi_W[J]}{\Phi_W[J_{\rm tar}]}-1
\right|
=
\left\|\frac{J}{J_{\rm tar}}-1\right\|_{\infty,\Omega_R} .
\label{eq:SM-operational-duality}
\end{equation}
\end{proposition}

\begin{proof}
Set
\begin{equation}
r_{\rm rel}(\omega)=\frac{J(\omega)}{J_{\rm tar}(\omega)}-1,
\qquad
\dd\nu_W(\omega)=
\frac{W(\omega)J_{\rm tar}(\omega)\dd\omega}
{\int_{\Omega_R}WJ_{\rm tar}\dd\omega} .
\end{equation}
By the definitions above, $\Phi_W[J]/\Phi_W[J_{\rm tar}]-1=\int_{\Omega_R}r_{\rm rel}\,\dd\nu_W$, whose absolute value is bounded by $\|r_{\rm rel}\|_\infty$. Conversely, choosing a sequence of nonnegative continuous bump functions whose weights concentrate near a frequency at which $|r_{\rm rel}|$ attains its supremum makes the absolute values of the corresponding weighted averages approach $\|r_{\rm rel}\|_\infty$.
\end{proof}

\subsection{Spectral estimates and mode-count lower bounds within the specified model class}

Let $\widehat J(\omega)>0$ be a continuous spectral estimate over the entire two-sided frequency band. Define its uniform relative spectral mismatch with respect to an amplitude-adjustable reference spectrum $\mathcal A/|\omega|$ by
\begin{equation}
\delta_R[\widehat J]=
\inf_{\mathcal A>0}
\sup_{\omega\in\Omega_R}
\left|1-\frac{|\omega|\widehat J(\omega)}{\mathcal A}\right| .
\label{eq:SM-data-defect}
\end{equation}
If
\begin{equation}
Y_-=\inf_{\Omega_R}|\omega|\widehat J(\omega),
\qquad
Y_+=\sup_{\Omega_R}|\omega|\widehat J(\omega),
\end{equation}
then balancing the upper and lower extrema gives
\begin{equation}
\mathcal A_R^\star=\frac{Y_++Y_-}{2},
\qquad
\delta_R[\widehat J]=\frac{Y_+-Y_-}{Y_++Y_-} .
\label{eq:SM-data-defect-closed}
\end{equation}
Below we abbreviate $\delta_R:=\delta_R[\widehat J]$. If a candidate network with $n$ auxiliary modes satisfies
\begin{equation}
\eta=
\sup_{\omega\in\Omega_R}
\left|\frac{J(\omega)}{\widehat J(\omega)}-1\right|<1,
\end{equation}
then its maximum relative error with respect to $\mathcal A_R^\star/|\omega|$ over $\Omega_R$ is at most
\begin{equation}
\varepsilon_{\rm eff}
=\delta_R+\eta+\delta_R\eta .
\label{eq:SM-effective-error}
\end{equation}
After rescaling the coupling vector as
$\widetilde{\bm g}=\bm g/\sqrt{\mathcal A_R^\star}$,
we obtain
$\widetilde J=J/\mathcal A_R^\star$,
whereas $H$, $\Gamma$, and the number of auxiliary modes remain unchanged.
The main theorem therefore applies directly to the normalized target $1/|\omega|$.
When $0<\varepsilon_{\rm eff}<1$, the actual number of auxiliary modes in the candidate network satisfies
\begin{equation}
n\ge\Nmin(R,\varepsilon_{\rm eff}) .
\label{eq:SM-data-mode-bound}
\end{equation}
This bound concerns the minimum number of auxiliary modes within the specified model class, rather than the number of microscopic defects or the dimension of the environmental Hilbert space~\cite{Vieira2024}.

\subsection{Simultaneous confidence bands and multiband mode-count lower bounds}

Suppose that an experiment yields a simultaneous confidence band
\begin{equation}
0<J_-(\omega)\le J_{\rm true}(\omega)\le J_+(\omega).
\end{equation}
Define
\begin{align}
Y_-^{\rm band}&=\inf_{\Omega_R}|\omega|J_-(\omega),\\
Y_+^{\rm band}&=\sup_{\Omega_R}|\omega|J_+(\omega),\\
\mathcal A_R^{\rm band}&=\frac{Y_+^{\rm band}+Y_-^{\rm band}}{2},\\
\delta_R^{\rm band}&=
\frac{Y_+^{\rm band}-Y_-^{\rm band}}
{Y_+^{\rm band}+Y_-^{\rm band}} .
\end{align}
If the spectrum of a candidate network lies within this simultaneous confidence band at every $\omega\in\Omega_R$, then $Y_-^{\rm band}\le |\omega|J(\omega)\le Y_+^{\rm band}$.
Thus, the maximum relative error of the network with respect to the reference spectrum $\mathcal A_R^{\rm band}/|\omega|$ over the target band is no greater than $\delta_R^{\rm band}$, and we may set $\varepsilon_{\rm eff}=\delta_R^{\rm band}$.

Alternatively, if the true spectrum $J_{\rm true}$ lies within the simultaneous confidence band and the fitting error of a candidate network relative to $J_{\rm true}$ satisfies $\eta_{\rm true}:=\sup_{\omega\in\Omega_R}\left|\frac{J(\omega)}{J_{\rm true}(\omega)}-1\right|<1$, we may set $\varepsilon_{\rm eff}=\delta_R^{\rm band}+\eta_{\rm true}+\delta_R^{\rm band}\eta_{\rm true}$.

If the true spectrum exhibits approximate $1/|\omega|$ scaling only over certain subbands, choose a family of two-sided subbands
\begin{equation}
\mathscr B\subseteq
\left\{
[-\omega_b,-\omega_a]\cup[\omega_a,\omega_b]:
1\le\omega_a<\omega_b\le R
\right\}.
\end{equation}
For any
$\mathcal I=[-\omega_b,-\omega_a]\cup[\omega_a,\omega_b]
\in\mathscr B$,
define $R_{\mathcal I}:=\frac{\omega_b}{\omega_a}$. Under the rescaling $\widetilde\omega=\omega/\omega_a$, this subband is mapped to $[-R_{\mathcal I},-1]\cup[1,R_{\mathcal I}]$, and this frequency rescaling does not change the number of auxiliary modes in the candidate network. Restrict the preceding spectral mismatch and fitting error to $\mathcal I$, and define $\varepsilon_{{\rm eff},\mathcal I}$ in the same way. The actual number of auxiliary modes in the same candidate network must then satisfy
\begin{equation}
n\ge
\sup_{\mathcal I\in\mathscr B}
\Nmin\!\left(R_{\mathcal I},
\varepsilon_{{\rm eff},\mathcal I}\right) .
\label{eq:SM-multiscale-bound}
\end{equation}
If confidence bands are constructed separately for each subband, multiple comparisons must also be controlled; a simultaneous confidence band covering the entire frequency band can be restricted directly to each subband.

\subsection{Power-law exponent mismatch}

For a pure power-law candidate spectrum
$\widehat J_\alpha(\omega)=\mathcal A_\alpha/|\omega|^\alpha$,
where $\mathcal A_\alpha>0$ and $\alpha\in\mathbb R$, substituting it into Eq.~\eqref{eq:SM-data-defect} and minimizing over the reference amplitude $\mathcal A>0$ (equivalently, optimizing the overall amplitude ratio $\mathcal A_\alpha/\mathcal A$) gives
\begin{equation}
\delta_\alpha(R)=
\frac{R^{|1-\alpha|}-1}{R^{|1-\alpha|}+1}
=\tanh\!\left(\frac{|1-\alpha|\ln R}{2}\right) .
\label{eq:SM-power-mismatch}
\end{equation}
For a prescribed mismatch tolerance $0\le\delta<1$, the condition
$\delta_\alpha(R)\le\delta$ is equivalent to
\begin{equation}
|1-\alpha|\le
\frac{2\operatorname{artanh}\delta}{\ln R} .
\end{equation}
Figure~\ref{fig:SM-alpha} shows how this mismatch grows with the log-frequency span and the deviation of the exponent from unity.

\begin{figure}[!htbp]
\centering
\includegraphics[width=0.80\textwidth]{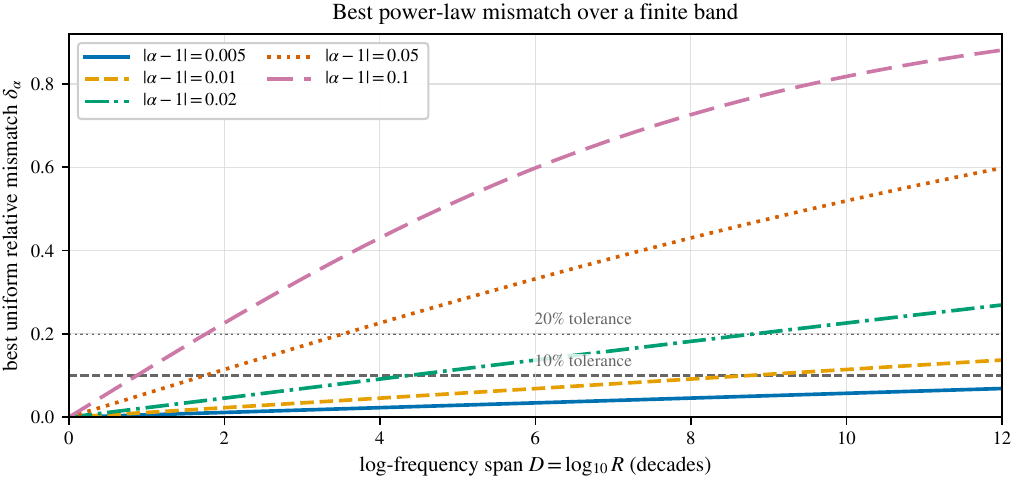}
\caption[Best power-law mismatch over a finite band]{Best uniform relative spectral mismatch between the candidate power-law scaling $1/|\omega|^\alpha$ and the reference scaling $1/|\omega|$ after optimization over the overall amplitude ratio. The exact expression is given in Eq.~\eqref{eq:SM-power-mismatch}.}
\label{fig:SM-alpha}
\end{figure}
\FloatBarrier

\section{Independent analytical and numerical cross-checks}\label{sec:SM-validation}

This section independently checks the elliptic moduli, rational degrees, scaling factors, and extrema over the continuous frequency band, as well as the numerical implementation of the corresponding formulas. The main tests comprise the Jacobi product formula and modular equation in $35$ cases, with agreement to at least $70$ significant digits; the minimum number of auxiliary modes in $84$ cases, with agreement in $84/84$ cases; and direct nonconvex optimization in $9$ cases, with a maximum relative discrepancy of $7.47\times10^{-5}$.

\subsection{High-precision numerical verification of the analytical formulas}

We take
\begin{equation}
R\in\{1.2,2,3,10,10^2,10^3,10^6\},
\qquad
N\in\{1,2,3,4,8\},
\end{equation}
giving $35$ test cases in total. The alternation points are given analytically by
\begin{equation}
x_j=\operatorname{dn}^{-2}
\!\left(\frac{jK'}{2N};k'\right),
\qquad j=0,\ldots,2N,
\label{eq:SM-alternation-points}
\end{equation}
where $\operatorname{dn}(u;k')$ is a Jacobi elliptic function, with $x_0=1$ and $x_{2N}=R^2$. The optimal errors obtained independently from the Jacobi product formula and the modular equation agree to at least $70$ significant digits; both the equioscillation relations and the vanishing of the derivatives at the interior alternation points are verified to the same precision. We also take $R=10^D$, $D=1,\ldots,12$, and
$\varepsilon\in\{0.2,0.1,0.05,0.02,0.01,10^{-3},10^{-4}\}$.
This gives $84$ test cases for the minimum number of auxiliary modes, for all of which the closed-form expression agrees with a sequential search over the mode count.

\subsection{Direct optimization without using the closed-form pole locations}

For independent damped auxiliary modes at zero detuning in the uncoupled diagonal subclass, the spectrum can be written as
\begin{equation}
J_N(\omega)=
\sum_{j=1}^{N}\frac{a_j}{\omega^2+\gamma_j^2},
\qquad a_j\ge0.
\end{equation}
Let $x=\omega^2$, and define the continuous error function used in the optimization by
\begin{equation}
\rho_N(x)=
1-\sqrt{x}\sum_{j=1}^{N}\frac{a_j}{x+\gamma_j^2},
\qquad x\in[1,R^2].
\label{eq:SM-rho}
\end{equation}
Its maximum absolute value on $[1,R^2]$ is exactly the maximum relative error of the corresponding network over the target band, which is the optimization objective in this section. For fixed pole locations, a linear program computes the minimax error level on the current exchange set; this value is a lower bound for the fixed-pole minimax problem over the full interval. Evaluating the resulting approximation over the entire interval and taking the supremum gives a continuous-error upper bound for the current candidate parameter set. The algorithm repeatedly finds all real roots of $\rho_N'(x)=0$ in $[1,R^2]$ and adds to the exchange set any extremum at which $|\rho_N(x)|$ exceeds the current minimax error level on the exchange set, until the gap between the upper and lower bounds on the continuous error is below $10^{-11}$. The outer optimization uses
$\ln\gamma_j\in[-\ln(20R),\ln(20R)]$ as its variables and employs the differential evolution algorithm~\cite{StornPrice1997}; the population size is $8N$, with at most $50$ generations for $N<3$ and at most $70$ generations for $N=3$, followed by up to $1400$ steps of local refinement using the Nelder--Mead algorithm~\cite{NelderMead1965}. For the nine problems with
$R\in\{3,10^2,10^3\}$ and $N=1,2,3$, the maximum relative discrepancy between the direct optimization and the exact values is $7.47\times10^{-5}$. The analytical proofs in Secs.~\ref{sec:SM-degree}--\ref{sec:SM-realization} establish the globally optimal values; the direct optimization here provides an independent numerical cross-check.

\end{document}